\documentclass{siamltex}

\usepackage{amsmath, amssymb}
\usepackage{algorithm}
\usepackage{algorithmic}
\usepackage{url}
\usepackage{graphicx}
\usepackage{footnote}
\usepackage{multirow}
\usepackage{color}
\usepackage{paralist}
\usepackage{hyperref}

\makesavenoteenv{tabular}

\newcommand{\CG}{\textrm{\tiny CG}}

\newcommand{\T}{^{\raisebox{0.09em}{\hbox{\tiny $T$}}\hspace{-0.1em}}}

\title{LSRN: A Parallel Iterative Solver for Strongly
    \\ Over- or Under-Determined Systems}

\author{Xiangrui Meng%
  \thanks{ICME, Stanford University, Stanford, CA 94305 (mengxr@stanford.edu).
    Partially supported by the U.S. Army Research Laboratory, through the
    Army High Performance Computing Research Center, Cooperative Agreement W911NF-07-0027
    and by NSF grant DMS-1009005.}
  \and Michael A. Saunders%
  \thanks{Systems Optimization Laboratory, Department of Management Science and
    Engineering, Stanford University, Stanford CA 94305 (saunders@stanford.edu).
    Partially supported by the U.S. Army Research Laboratory, through the
    Army High Performance Computing Research Center, Cooperative Agreement W911NF-07-0027
    and by NSF grant DMS-1009005.}
  \and Michael W. Mahoney%
  \thanks{Department of Mathematics, Stanford University, CA 94305 (mmahoney@cs.stanford.edu).
    Partially supported by NSF grant DMS-1009005.}}

\begin{document}

\maketitle

\begin{abstract}
  We describe a parallel iterative least squares solver named \texttt{LSRN} that
  is based on random normal projection.  \texttt{LSRN} computes the
  min-length solution to $\min_{x \in \mathbb{R}^n} \|A x - b\|_2$, where $A \in
  \mathbb{R}^{m \times n}$ with $m \gg n$ or $m \ll n$, and where $A$ may be
  rank-deficient. Tikhonov regularization may also be included. Since $A$ is
  only involved in matrix-matrix and matrix-vector multiplications, it can be a
  dense or sparse matrix or a linear operator, and \texttt{LSRN} automatically
  speeds up when $A$ is sparse or a fast linear operator. The preconditioning
  phase consists of a random normal projection, which is embarrassingly
  parallel, and a singular value decomposition of size $\lceil \gamma \min(m,n)
  \rceil \times \min(m,n)$, where $\gamma$ is moderately larger than $1$, e.g.,
  $\gamma = 2$.  We prove that the preconditioned system is well-conditioned,
  with a strong concentration result on the extreme singular values, and hence
  that the number of iterations is fully predictable when we apply LSQR or the
  Chebyshev semi-iterative method.  As we demonstrate, the Chebyshev method is
  particularly efficient for solving large problems on clusters with high
  communication cost.  Numerical results demonstrate that on a shared-memory
  machine, \texttt{LSRN} outperforms LAPACK's DGELSD on large dense problems,
  and MATLAB's backslash (SuiteSparseQR) on sparse problems. Further experiments
  demonstrate that \texttt{LSRN} scales well on an Amazon Elastic Compute Cloud
  cluster.
\end{abstract}



\begin{DOI}
   xxx/xxxxxxxxx
\end{DOI}

\section{Introduction}
\label{sec:introduction}

Randomized algorithms have become indispensable in many areas of computer
science, with applications ranging from complexity theory to combinatorial
optimization, cryptography, and machine learning.  Randomization has also been
used in numerical linear algebra (for instance, the initial vector in the power
iteration is chosen at random so that almost surely it has a nonzero component
along the direction of the dominant eigenvector), yet most well-developed matrix
algorithms, e.g., matrix factorizations and linear solvers, are deterministic.
In recent years, however, motivated by large data problems, very
nontrivial randomized algorithms for very large matrix problems have drawn
considerable attention from researchers, originally in theoretical computer
science and subsequently in numerical linear algebra and scientific computing.
By randomized algorithms, we refer in particular to random sampling and random
projection algorithms
\cite{drineas2006sampling,sarlos2006improved,drineas2007faster,rokhlin2008fast,avron2010blendenpik}.
For a comprehensive overview of these developments, see the review of
Mahoney~\cite{Mah-mat-rev_BOOK}, and for an excellent overview of numerical
aspects of coupling randomization with classical low-rank matrix factorization
methods, see the review of Halko, Martinsson, and Tropp~\cite{halko2011finding}.

Here, we consider high-precision solving of linear least squares (LS) problems
that are strongly over- or under-determined, and possibly rank-deficient.  In
particular, given a matrix $A \in \mathbb{R}^{m \times n}$ and a vector $b \in
\mathbb{R}^m$, where $m \gg n$ or $m \ll n$ and we do not assume that $A$ has
full rank, we wish to develop randomized algorithms to compute accurately the
unique min-length solution to the problem
\begin{equation}
  \label{eq:ls}
  \text{minimize}_{x \in \mathbb{R}^n}  \quad \| A x - b \|_2  .
\end{equation}
If we let $r = \rank(A) \leq \min(m,n)$, then recall that if $r < n$ (the LS
problem is under-determined or rank-deficient), then \eqref{eq:ls} has an
infinite number of minimizers.  In that case, the set of all minimizers is
convex and hence has a unique element having minimum length.  On the other hand,
if $r = n$ so the problem has full rank, there exists only one minimizer
to \eqref{eq:ls} and hence it must have the minimum length.  In either case, we
denote this unique min-length solution to \eqref{eq:ls} by $x^*$.  That is,
\begin{eqnarray}
  \label{eq:ls_min_length}
  x^* = \arg \min \| x \|_2 \quad \text{subject to} \quad x \in \arg \min_z \| A z - b \|_2.
\end{eqnarray}
LS problems of this form have a long history, tracing back to Gauss, and they
arise in numerous applications.  The demand for faster LS solvers will continue
to grow in light of new data applications and as problem scales become larger
and larger.

In this paper, we describe an LS solver called \texttt{LSRN} for these strongly
over- or under-determined, and possibly rank-deficient, systems. \texttt{LSRN}
uses random normal projections to compute a preconditioner matrix such that the
preconditioned system is provably extremely well-conditioned.  Importantly for
large-scale applications, the preconditioning process is embarrassingly
parallel, and it automatically speeds up with sparse matrices and fast linear
operators.  LSQR~\cite{paige1982lsqr} or the Chebyshev semi-iterative (CS)
method~\cite{golub1961chebyshev} can be used at the iterative step to compute
the min-length solution within just a few iterations.  We show that the latter
method is preferred on clusters with high communication cost.

Because of its provably-good conditioning properties, \texttt{LSRN} has a fully
predictable run-time performance, just like direct solvers, and it scales well
in parallel environments. On large dense systems, \texttt{LSRN} is faster than
LAPACK's DGELSD for strongly over-determined problems, and is much faster for
strongly under-determined problems, although solvers using fast random
projections, like Blendenpik~\cite{avron2010blendenpik}, are still slightly
faster in both cases.  On sparse systems, \texttt{LSRN} runs significantly
faster than competing solvers, for both the strongly over- or under-determined
cases.

In section \ref{sec:linear-least-squares} we describe existing deterministic LS
solvers and recent randomized algorithms for the LS problem.  In section
\ref{sec:prec-line-least} we show how to do preconditioning correctly for
rank-deficient LS problems, and in section \ref{sec:prec-via-rand} we introduce
\texttt{LSRN} and discuss its properties.  Section \ref{sec:regularization}
describes how \texttt{LSRN} can handle Tikhonov regularization for both over-
and under-determined systems, and in section \ref{sec:experiments} we provide a
detailed empirical evaluation illustrating the behavior of \texttt{LSRN}.

\section{Least squares solvers}
\label{sec:linear-least-squares}

In this section we discuss related work, including deterministic direct and
iterative methods as well as recently developed randomized methods, for
computing solutions to LS problems, and we discuss how our results fit into this
broader context.

\subsection{Deterministic methods}
\label{sec:deter-method}

It is well known that $x^*$ in \eqref{eq:ls_min_length} can be computed using
the singular value decomposition (SVD) of $A$. Let $A = U \Sigma V\T$ be the
economy-sized SVD, where $U \in \mathbb{R}^{m \times r}$, $\Sigma \in
\mathbb{R}^{r \times r}$, and $V \in \mathbb{R}^{n \times r}$. We have $x^* = V
\Sigma^{-1} U\T b$.  The matrix $V \Sigma^{-1} U\T$ is the Moore-Penrose
pseudoinverse of $A$, denoted by $A^\dagger$. The pseudoinverse is defined and
unique for any matrix. Hence we can simply write $x^* = A^\dagger b$. The SVD
approach is accurate and robust to rank-deficiency.

Another way to solve \eqref{eq:ls_min_length} is using a complete orthogonal
factorization of $A$. If we can find orthonormal matrices $Q \in \mathbb{R}^{m
  \times r}$ and $Z \in \mathbb{R}^{n \times r}$, and a matrix $T \in
\mathbb{R}^{r \times r}$, such that $A = Q T Z\T$, then the min-length solution
is given by $x^* = Z T^{-1} Q\T b$. We can treat SVD as a special case of
complete orthogonal factorization. In practice, complete orthogonal
factorization is usually computed via rank-revealing QR factorizations, making
$T$ a triangular matrix. The QR approach is less expensive than SVD, but it is
slightly less robust at determining the rank of $A$.

A third way to solve \eqref{eq:ls_min_length} is by computing the min-length
solution to the normal equation $A\T A x = A\T b$, namely
\begin{equation}
  \label{eq:ls_ne_min_length}
  x^* = (A\T A)^\dagger A\T b = A\T (A A\T)^\dagger b.
\end{equation}
It is easy to verify the correctness of \eqref{eq:ls_ne_min_length} by replacing
$A$ by its economy-sized SVD $U \Sigma V\T$. If $r = \min(m,n)$, a Cholesky
factorization of either $A\T A$ (if $m \geq n$) or $AA\T$ (if $m \leq n$) solves
\eqref{eq:ls_ne_min_length} nicely. If $r < \min(m,n)$, we need the eigensystem
of $A\T A$ or $A A\T$ to compute $x^*$.  The normal equation approach is the
least expensive among the three direct approaches we have mentioned, especially
when $m \gg n$ or $m \ll n$, but it is also the least accurate one, especially
on ill-conditioned problems.  See Chapter~5 of Golub and Van Loan
\cite{golub1996matrix} for a detailed analysis.

Instead of these direct methods, we can use iterative methods to solve
\eqref{eq:ls}.  If all the iterates $\{ x^{(k)} \}$ are in $\text{range}(A\T)$
and if $\{ x^{(k)} \}$ converges to a minimizer, it must be the minimizer having
minimum length, i.e., the solution to \eqref{eq:ls_min_length}.  This is the
case when we use a Krylov subspace method starting with a zero vector.  For
example, the conjugate gradient (CG) method on the normal equation leads to the
min-length solution (see Paige and Saunders~\cite{paige1975solution}).  In
practice, CGLS~\cite{hestenesmethods}, LSQR~\cite{paige1982lsqr} are preferable
because they are equivalent to applying CG to the normal equation in exact
arithmetic but they are numerically more stable. Other Krylov subspace methods
such as the CS method~\cite{golub1961chebyshev}
and LSMR~\cite{fong2011lsmr} can solve \eqref{eq:ls} as well.

Importantly, however, it is in general hard to predict the number of iterations
for CG-like methods.  The convergence rate is affected by the condition number
of $A\T A$.  A classical result \cite[p.187]{luenberger1973introduction} states
that
\begin{equation}
  \label{eq:cg_convergence_rate}
  \frac{\| x^{(k)} - x^* \|_{A\T A}}{\|x^{(0)} - x^*\|_{A\T A}} \leq 2 \left( \frac{\sqrt{\kappa(A\T A)} - 1}{\sqrt{\kappa(A\T A)} + 1} \right)^k,
\end{equation}
where $\| z \|_{A\T A} = z\T A\T A z = \|A z\|^2$ for any $z \in \mathbb{R}^n$,
and where $\kappa(A\T A)$ is the condition number of $A\T A$ under the $2$-norm.
Estimating $\kappa(A\T A)$ is generally as hard as solving the LS problem
itself, and in practice the bound does not hold in any case unless
reorthogonalization is used.  Thus, the computational cost of CG-like methods
remains unpredictable in general, except when $A\T A$ is very well-conditioned
and the condition number can be well estimated.

\subsection{Randomized methods}


In 2007, Drineas, Mahoney, Muthukrishnan, and
Sarl{\'o}s~\cite{drineas2007faster} introduced two randomized algorithms for the
LS problem, each of which computes a relative-error approximation to the
min-length solution in $\mathcal{O}(m n \log n)$ time, when $m \gg n$.  Both of
these algorithms apply a randomized Hadamard transform to the columns of $A$,
thereby generating a problem of smaller size, one using uniformly random
sampling and the other using a sparse random projection.  They proved that, in
both cases, the solution to the smaller problem leads to relative-error
approximations of the original problem.  The accuracy of the approximate
solution depends on the sample size; and to have relative precision
$\varepsilon$, one should sample $\mathcal{O}(n/\varepsilon)$ rows after the
randomized Hadamard transform.  This is suitable when low accuracy is
acceptable, but the $\varepsilon$ dependence quickly becomes the bottleneck
otherwise.  Using those algorithms as preconditioners was also mentioned in
\cite{drineas2007faster}. This work laid the ground for later algorithms and
implementations.


Later, in 2008, Rokhlin and Tygert \cite{rokhlin2008fast} described a related
randomized algorithm for over-determined systems.  They used a randomized
transform named SRFT that consists of $m$ random Givens rotations, a random
diagonal scaling, a discrete Fourier transform, and a random sampling.  They
considered using their method as a preconditioning method, and they showed that
to get relative precision $\varepsilon$, only $\mathcal{O}(n \log (1 /
\varepsilon))$ samples are needed.  In addition, they proved that if the sample
size is greater than $4 n^2$, the condition number of the preconditioned system
is bounded above by a constant.  Although choosing this many samples would
adversely affect the running time of their solver, they also illustrated
examples of input matrices for which the $4 n^2$ sample bound was weak and for
which many fewer samples sufficed.


Then, in 2010, Avron, Maymounkov, and Toledo~\cite{avron2010blendenpik}
implemented a high-precision LS solver, called Blendenpik, and compared it to
LAPACK's DGELS and to LSQR with no preconditioning.  Blendenpik uses a
Walsh-Hadamard transform, a discrete cosine transform, or a discrete Hartley
transform for blending the rows/columns, followed by a random sampling, to
generate a problem of smaller size.  The $R$ factor from the QR factorization of
the smaller matrix is used as the preconditioner for LSQR.  Based on their
analysis, the condition number of the preconditioned system depends on the
coherence or statistical leverage scores of $A$, i.e., the maximal row norm of
$U$, where $U$ is an orthonormal basis of $\text{range}(A)$.  We note that a
solver for under-determined problems is also included in the Blendenpik package.


In 2011, Coakley, Rokhlin, and Tygert \cite{coakley2011fast} described an
algorithm that is also based on random normal projections. It computes the
orthogonal projection of any vector $b$ onto the null space of $A$ or onto the
row space of $A$ via a preconditioned normal equation.  The algorithm solves the
over-determined LS problem as an intermediate step. They show that the normal
equation is well-conditioned and hence the solution is reliable. For an
over-determined problem of size $m \times n$, the algorithm requires applying
$A$ or $A\T$\, $3 n + 6$ times, while \texttt{LSRN} needs approximately $2 n +
200$ matrix-vector multiplications under the default setting.
Asymptotically, \texttt{LSRN} will become faster as $n$ increases beyond
several hundred.
See section \ref{subsec:complexity} for further complexity analysis of
\texttt{LSRN}.

\subsection{Relationship with our contributions}

All prior approaches assume that $A$ has full rank, and for those based on
iterative solvers, none provides a tight upper bound on the condition number of
the preconditioned system (and hence the number of iterations). For
\texttt{LSRN}, Theorem~\ref{thm:ls_precond_sufficient} ensures that the
min-length solution is preserved, independent of the rank, and
Theorems~\ref{thm:cond_bound} and \ref{thm:iter} provide bounds on the condition
number and number of iterations, independent of the spectrum of $A$.  In
addition to handling rank-deficiency well, \texttt{LSRN} can even take advantage
of it, resulting in a smaller condition number and fewer iterations.

Some prior work on the LS problem has explored ``fast'' randomized transforms
that run in roughly $\mathcal{O}( m n \log m )$ time on a dense matrix $A$,
while the random normal projection we use in \texttt{LSRN} takes $\mathcal{O}( m
n^2 )$ time.  Although this could be an issue for some applications, the use of
random normal projections comes with several advantages.  First, if $A$ is a
sparse matrix or a linear operator, which is common in large-scale applications,
then the Hadamard-based fast transforms are no longer ``fast''. Second, the
random normal projection is easy to implement using threads or MPI, and it
scales well in parallel environments.  Third, the strong symmetry of the
standard normal distribution helps give the strong high probability bounds on
the condition number in terms of sample size. These bounds depend on nothing but
$s/r$, where $s$ is the sample size.  For example, if $s = 4r$,
Theorem~\ref{thm:cond_bound} ensures that, with high probability, the condition
number of the preconditioned system is less than $3$.

This last property about the condition number of the preconditioned system makes
the number of iterations and thus the running time of \texttt{LSRN} fully
predictable like for a direct method. It also enables use of the
CS method, which needs only one level-1 and two level-2
BLAS operations per iteration, and is particularly suitable for clusters with
high communication cost because it doesn't have vector inner products that
require synchronization between nodes.  Although the CS method has the same
theoretical upper bound on the convergence rate as CG-like methods, it requires
accurate bounds on the singular values in order to work efficiently.  Such
bounds are generally hard to come by, limiting the popularity of the CS method
in practice, but they are provided for the preconditioned system by our
Theorem~\ref{thm:cond_bound}, and we do achieve high efficiency in our
experiments.

\section{Preconditioning for linear least squares}
\label{sec:prec-line-least}

In light of \eqref{eq:cg_convergence_rate}, much effort has been made to
transform a linear system into an equivalent system with reduced condition
number. This \emph{preconditioning}, for a square linear system $B x = d$ of
full rank, usually takes one of the following forms:
\begin{align*}
  \label{eq:linear_precond}
  \text{left preconditioning} &\quad M\T B x = M\T d, \\
  \text{right preconditioning} &\quad B N y = d, ~ x = N y, \\
  \text{left and right preconditioning} &\quad M\T B N y = M\T d, ~ x = N y.
\end{align*}
Clearly, the preconditioned system is consistent with the original one, i.e.,
has the same $x^*$ as the unique solution, if the preconditioners $M$ and $N$
are nonsingular.

For the general LS problem \eqref{eq:ls_min_length}, preconditioning needs
better handling in order to produce the same min-length solution as the original
problem.  For example, if we apply left preconditioning to the LS problem
$\min_x \| A x - b \|_2$, the preconditioned system becomes $\min_x \| M\T A x -
M\T b \|_2$, and its min-length solution is given by
\begin{equation*}
  x^*_{\text{left}} = (M\T A)^\dagger M\T b.
\end{equation*}
Similarly, the min-length solution to the right preconditioned system is 
given by
\begin{equation*}
  x^*_{\text{right}} = N ( A N )^\dagger b.
\end{equation*}
The following lemma states the necessary and sufficient conditions for $A^\dagger
= N ( A N )^\dagger$ or $A^\dagger = ( M\T A )^\dagger M\T$ to hold. Note that
these conditions holding certainly imply that $x_{\text{right}}^* = x^*$ and
$x_{\text{left}}^* = x^*$, respectively.

\begin{lemma}
  \label{lemma:ls_precond}
  Given $A \in \mathbb{R}^{m \times n}$, $N \in \mathbb{R}^{n \times p}$ and $M \in
  \mathbb{R}^{m \times q}$, we have
  \begin{enumerate}
  \item $A^\dagger = N ( A N )^\dagger$ if and only if $\text{range}(N N\T A\T )
    = \text{range} ( A\T )$,
  \item $A^\dagger = ( M\T A )^\dagger M\T$ if and only if $\text{range}(M M\T
    A) = \text{range}(A)$.
  \end{enumerate}
\end{lemma}

\begin{proof} Let $r = \rank(A)$ and $U \Sigma V\T$ be $A$'s economy-sized SVD
  as in section \ref{sec:deter-method}, with $A^\dagger = V \Sigma^{-1}
  U\T$. Before continuing our proof, we reference the following facts about the
  pseudoinverse:
  \begin{enumerate}
  \item $B^\dagger = B\T ( B B\T )^\dagger$ for any matrix $B$,
  \item For any matrices $B$ and $C$ such that $B C$ is defined, $(B C)^\dagger
    = C^\dagger B^\dagger$ if
    \begin{inparaenum}[(i)]
    \item $B\T B = I$ or
    \item $C C\T = I$ or
    \item $B$ has full column rank and $C$ has full row rank.
    \end{inparaenum}
  \end{enumerate}
  Now let's prove the ``if'' part of the first statement. If $\text{range}(N N\T
  A\T ) = \text{range} ( A\T ) = \text{range}(V)$, we can write $N N\T A\T$ as
  $V Z$ where $Z$ has full row rank. Then,
  \begin{eqnarray*}
    N ( A N )^\dagger &=& N ( A N )\T ( A N ( A N )\T )^\dagger = N N\T A\T  ( A N N\T A\T )^\dagger \\ &=& V Z ( U \Sigma V\T V Z )^\dagger = V Z ( U \Sigma Z )^\dagger = V Z Z^\dagger \Sigma^{-1} U\T = V \Sigma^{-1} U\T = A^\dagger.
  \end{eqnarray*}
  Conversely, if $N(AN)^\dagger = A^\dagger$, we know that $\text{range}( N ( A
  N )^\dagger ) = \text{range} (A^\dagger) = \text{range}( V )$ and hence
  $\text{range}( V ) \subseteq \text{range}( N )$. Then we can decompose $N$ as
  $ ( V \ V_c ) \hbox{\scriptsize $\begin{pmatrix} Z \\ Z_c
  \end{pmatrix}$}= V Z + V_c Z_c$, where $V_c$ is orthonormal, $V\T V_c = 0$, and
  {\scriptsize $\smash[tb]{
  \begin{pmatrix}
    Z \\ Z_c 
  \end{pmatrix}}$} has full row rank. Then,
  \begin{eqnarray*}
    0 &=& N ( A N )^\dagger - A^\dagger = ( V Z + V_c Z_c ) ( U \Sigma V\T ( V Z + V_c Z_c ) )^\dagger - V \Sigma^{-1} U\T \\
    &=& ( V Z + V_c Z_c ) ( U \Sigma Z )^\dagger - V \Sigma^{-1} U\T \\
    &=& ( V Z + V_c Z_c ) Z^\dagger \Sigma^{-1} U\T - V \Sigma^{-1} U\T = V_c Z_c Z^\dagger \Sigma^{-1} U\T.
  \end{eqnarray*}
  Multiplying by $V_c\T$ on the left and $U \Sigma$ on the right, we get $Z_c
  Z^\dagger = 0$, which is equivalent to $Z_c Z\T = 0$. Therefore,
  \begin{eqnarray*}
    \text{range}( N N\T A\T ) &=& \text{range}( ( V Z + V_c Z_c ) ( V Z + V_c Z_c )\T V \Sigma U\T ) \\
    &=& \text{range}( ( V Z Z\T V\T + V_c Z_c Z_c\T V_c\T ) V \Sigma U\T ) \\
    &=& \text{range} ( V Z Z\T \Sigma U\T ) \\
    &=& \text{range} ( V ) = \text{range} ( A\T ),
  \end{eqnarray*}
  where we used the facts that $Z$ has full row rank and hence $Z Z\T$ is nonsingular,
  $\Sigma$ is nonsingular, and $U$ has full column rank.

  To prove the second statement, let us take $B = A\T$. By the first statement,
  we know $B^\dagger = M ( B M )^\dagger$ if and only if $\text{range}( M M\T
  B\T ) = \text{range}(B\T)$, which is equivalent to saying $A^\dagger = ( M\T A
  )^\dagger M\T$ if and only if $\text{range}( M M\T A ) = \text{range} ( A )$.
\end{proof}

Although Lemma \ref{lemma:ls_precond} gives the necessary and sufficient
condition, it does not serve as a practical guide for preconditioning LS
problems.  In this work, we are more interested in a sufficient condition that
can help us build preconditioners.  To that end, we provide the following
theorem.

\begin{theorem}
  \label{thm:ls_precond_sufficient}
  Given $A \in \mathbb{R}^{m \times n}$, $b \in \mathbb{R}^m$, $N \in
  \mathbb{R}^{n \times p}$, and $M \in \mathbb{R}^{m \times q}$, let $x^*$ be
  the min-length solution to the LS problem $\min_x \| A x - b \|_2$,
  $x_{\text{right}}^* = N y^*$ where $y^*$ is the min-length solution to $\min_y
  \| A N y - b \|_2$, and $x_{\text{left}}^*$ be the min-length solution to
  $\min_x \| M\T A x - M\T b \|_2$.  Then,
  \begin{enumerate}
  \item $x_{\text{right}}^* = x^*$ if $\text{range}(N) = \text{range}(A\T)$,
  \item $x_{\text{left}}^* = x^*$ if $\text{range}(M) = \text{range}(A)$.
  \end{enumerate} 
\end{theorem}

\begin{proof} 
  Let $r = \rank(A)$ and $U \Sigma V\T$ be $A$'s economy-sized SVD. If
  $\text{range}(N) = \text{range}(A\T) = \text{range}(V)$, we can write $N$ as
  $V Z$, where $Z$ has full row rank. Therefore,
  \begin{eqnarray*}
    \text{range}( N N\T A\T ) &=& \text{range}( V Z Z\T V\T V \Sigma U\T ) = \text{range}( V Z Z\T \Sigma U\T ) \\
    &=& \text{range}( V ) = \text{range}(A\T).
  \end{eqnarray*}
  By Lemma \ref{lemma:ls_precond}, $A^\dagger = N ( A N )^\dagger$ and hence
  $x_{\text{left}}^* = x^*$. The second statement can be proved by similar
  arguments.
\end{proof}

\section{Algorithm \texttt{LSRN}}
\label{sec:prec-via-rand}

In this section we present \texttt{LSRN}, an iterative solver for
solving strongly over- or under-determined systems, based on ``random
normal projection''.  To construct a preconditioner we apply a
transformation matrix whose entries are independent random variables
drawn from the standard normal distribution.  We prove that the
preconditioned system is almost surely consistent with the original
system, i.e., both have the same min-length solution.  At least as
importantly, we prove that the spectrum of the preconditioned system
is independent of the spectrum of the original system; and we provide
a strong concentration result on the extreme singular values of the
preconditioned system.  This concentration result enables us to
predict the number of iterations for CG-like methods, and it also
enables use of the CS method, which requires an accurate bound on the
singular values to work efficiently.

\subsection{The algorithm}
\label{subsec:the_alg}

Algorithm \ref{alg:ls_randn_tall} shows the detailed procedure of \texttt{LSRN}
to compute the min-length solution to a strongly over-determined problem, and
Algorithm \ref{alg:ls_randn_fat} shows the detailed procedure for a strongly
under-determined problem.  We refer to these two algorithms together as
\texttt{LSRN}.  Note that they only use the input matrix $A$ for matrix-vector
and matrix-matrix multiplications, and thus $A$ can be a dense matrix, a sparse
matrix, or a linear operator.  In the remainder of this section we focus on
analysis of the over-determined case. We emphasize that analysis of the
under-determined case is quite analogous.

\begin{algorithm}
  \caption{\texttt{LSRN} (computes $\hat{x} \approx A^\dagger b$ when $m \gg n$)}
  \label{alg:ls_randn_tall}
  \begin{algorithmic}[1]
    \STATE Choose an oversampling factor $\gamma > 1$ and set $s = \lceil \gamma
    n \rceil$.

    \STATE Generate $G = \text{randn}(s,m)$, i.e., an $s$-by-$m$ random matrix
    whose entries are independent random variables following the standard normal
    distribution.

    \STATE Compute $\tilde{A} = G A$.

    \STATE Compute $\tilde{A}$'s economy-sized SVD $\tilde{U} \tilde{\Sigma}
    \tilde{V}\T$, where $r = \rank(\tilde{A})$, $\tilde{U} \in \mathbb{R}^{s
      \times r}$, $\tilde{\Sigma} \in \mathbb{R}^{r \times r}$, $\tilde{V} \in
    \mathbb{R}^{n \times r}$, and only $\tilde{\Sigma}$ and $\tilde{V}$ are
    needed.
    
    \STATE Let $N = \tilde{V} \tilde{\Sigma}^{-1}$.
  
    \STATE Compute the min-length solution to $\min_y \| A N y - b \|_2$ using
    an iterative method. Denote the solution by $\hat{y}$.

    \STATE Return $\hat{x} = N \hat{y}$.
  \end{algorithmic}
\end{algorithm}

\begin{algorithm}
  \caption{\texttt{LSRN} (computes $\hat{x} \approx A^\dagger b$ when $m \ll n$)}
  \label{alg:ls_randn_fat}
  \begin{algorithmic}[1]
    \STATE Choose an oversampling $\gamma > 1$ and set $s = \lceil \gamma m
    \rceil$.

    \STATE Generate $G = \text{randn}(n,s)$, i.e., an $n$-by-$s$ random matrix
    whose entries are independent random variables following the standard normal
    distribution.

    \STATE Compute $\tilde{A} = A G$.

    \STATE Compute $\tilde{A}$'s economy-sized SVD $\tilde{U} \tilde{\Sigma}
    \tilde{V}\T$, where $r = \rank(\tilde{A})$, $\tilde{U} \in \mathbb{R}^{n
      \times r}$, $\tilde{\Sigma} \in \mathbb{R}^{r \times r}$, $\tilde{V} \in
    \mathbb{R}^{s \times r}$, and only $\tilde{U}$ and $\tilde{\Sigma}$ are
    needed.
    
    \STATE Let $M = \tilde{U} \tilde{\Sigma}^{-1}$.
  
    \STATE Compute the min-length solution to $\min_x \| M\T A x - M\T b \|_2$
    using an iterative method, denoted by $\hat{x}$.

    \STATE Return $\hat{x}$.
  \end{algorithmic}
\end{algorithm}

\subsection{Theoretical properties}
\label{subsec:theoretical_properties}

The use of random normal projection offers \texttt{LSRN} some nice theoretical
properties. We start with consistency.

\begin{theorem}
  \label{thm:consistency}
  In Algorithm \ref{alg:ls_randn_tall}, we have $\hat{x} = A^\dagger b$ almost
  surely.
\end{theorem}

\begin{proof}
  Let $r = \rank(A)$ and $U \Sigma V\T$ be $A$'s economy-sized SVD. We have
  \begin{eqnarray*}
    \text{range}(N) &=& \text{range}(\tilde{V} \tilde{\Sigma}^{-1}) = \text{range}(\tilde{V}) \\
    &=& \text{range}(\tilde{A}\T) = \text{range}( A\T G\T ) = \text{range}( V \Sigma (GU)\T ).
  \end{eqnarray*}
  Define $G_1 = GU \in \mathbb{R}^{s \times r}$. Since $G$'s entries are
  independent random variables following the standard normal distribution and
  $U$ is orthonormal, $G_1$'s entries are also independent random variables
  following the standard normal distribution. Then given $s \geq \gamma n > n
  \geq r$, we know $G_1$ has full column rank $r$ with probability
  $1$. Therefore,
  \begin{equation*}
    \text{range}(N) = \text{range}( V \Sigma G_1\T ) = \text{range}(V) = \text{range}(A\T),
  \end{equation*}
  and hence by Theorem \ref{thm:ls_precond_sufficient} we have $\hat{x} =
  A^\dagger b$ almost surely.
\end{proof}

A more interesting property of \texttt{LSRN} is that the spectrum (the set of
singular values) of the preconditioned system is solely associated with a random
matrix of size $s \times r$, independent of the spectrum of the original system.

\begin{lemma}
  \label{lemma:spectrum}
  In Algorithm \ref{alg:ls_randn_tall}, the spectrum of $A N$ is the same as the
  spectrum of $G_1^\dagger = (G U)^\dagger$, independent of $A$'s spectrum.
\end{lemma}

\begin{proof}
  Following the proof of Theorem \ref{thm:consistency}, let $G_1 = U_1 \Sigma_1
  V_1\T$ be $G_1$'s economy-sized SVD, where $U_1 \in \mathbb{R}^{s \times r}$,
  $\Sigma_1 \in \mathbb{R}^{r \times r}$, and $V_1 \in \mathbb{R}^{r \times
    r}$. Since $\text{range}(\tilde{U}) = \text{range}(GA) = \text{range}(GU) =
  \text{range}(U_1)$ and both $\tilde{U}$ and $U_1$ are orthonormal matrices,
  there exists an orthonormal matrix $Q_1 \in \mathbb{R}^{r \times r}$ such that
  $U_1 = \tilde{U} Q_1$.  As a result,
  \begin{equation*}
    \tilde{U} \tilde{\Sigma} \tilde{V}\T = \tilde{A} = G U \Sigma V\T = U_1 \Sigma_1 V_1\T \Sigma V\T = \tilde{U} Q_1 \Sigma_1 V_1\T \Sigma V\T. 
  \end{equation*}
  Multiplying by $\tilde{U}\T$ on the left of each side, we get $\tilde{\Sigma}
  \tilde{V}\T = Q_1 \Sigma_1 V_1\T \Sigma V\T$. Taking the pseudoinverse gives
  $N = \tilde{V} \tilde{\Sigma}^{-1} = V \Sigma^{-1} V_1 \Sigma_{1}^{-1} Q_1\T$.
  Thus,
  \begin{equation*}
    A N= U \Sigma V\T  V \Sigma^{-1} V_1 \Sigma_{1}^{-1} Q_1\T = U V_1 \Sigma_{1}^{-1} Q_1\T\,,
  \end{equation*}
  which gives $AN$'s SVD. Therefore, $A N$'s singular values are
  $\text{diag}(\Sigma_1^{-1})$, the same as $G_1^\dagger$'s spectrum, but
  independent of $A$'s.
\end{proof}

We know that $G_1 = G U$ is a random matrix whose entries are independent random
variables following the standard normal distribution.  The spectrum of $G_1$ is
a well-studied problem in Random Matrix Theory, and in particular the properties
of extreme singular values have been studied.  Thus, the following lemma is
important for us. We use $\mathcal{P}(\cdot)$ to refer to the probability that a
given event occurs.

\begin{lemma} 
  \label{lemma:concentration}\textnormal{(Davidson and Szarek \cite{davidson2001local})} Consider an $s \times r$ random matrix
  $G_1$ with $s \geq r$, whose entries are independent random variables following the
  standard normal distribution. Let the singular values be $\sigma_1 \geq \cdots \geq
  \sigma_r$. Then for any $t > 0$,
  \begin{equation}
    \label{eq:concentration}
    \max \left\{ \mathcal{P}(\sigma_1 \geq \sqrt{s} + \sqrt{r} + t), \mathcal{P}( \sigma_r \leq \sqrt{s} - \sqrt{r} - t) \right\} < e^{- t^2 / 2}.
  \end{equation}
\end{lemma}

With the aid of Lemma \ref{lemma:concentration}, it is straightforward to 
obtain the concentration result of $\sigma_1(AN)$, $\sigma_r(AN)$, and 
$\kappa(AN)$ as follows.

\begin{theorem}
  \label{thm:cond_bound}
  In Algorithm \ref{alg:ls_randn_tall}, for any $\alpha \in ( 0, 1 - \sqrt{r/s} )$, we
  have
  \begin{equation}
    \label{eq:sgm_concentration}
    \mbox{$\max \left\{ \mathcal{P} \left(\sigma_1(AN) \geq \frac{1}{ (1-\alpha) \sqrt{s} - \sqrt{r}} \right), \mathcal{P} \left( \sigma_r(AN) \leq \frac{1}{ ( 1 + \alpha ) \sqrt{s} + \sqrt{r}} \right) \right\} < e^{-\alpha^2 s / 2}$}
  \end{equation}
and 
\begin{equation}
    \label{eq:cond_concentration}
    \mathcal{P} \left( \kappa(A N) = \frac{\sigma_1(AN)}{\sigma_r(AN)} \leq \frac{1 + \alpha + \sqrt{r/s}}{1 - \alpha - \sqrt{r/s}} \right) \geq 1 - 2 e^{-\alpha^2 s/2}.
  \end{equation}
\end{theorem}
\begin{proof}
  Set $t = \alpha \sqrt{s}$ in Lemma \ref{lemma:concentration}.
\end{proof}

In order to estimate the number of iterations for CG-like methods, we can now
combine \eqref{eq:cg_convergence_rate} and \eqref{eq:cond_concentration}.

\begin{theorem}
  \label{thm:iter}
  In exact arithmetic, given a tolerance $\varepsilon > 0$, a CG-like method
  applied to the preconditioned system $\min_y \| A N y - b \|_2$ with $y^{(0)}
  = 0$ converges within $(\log \varepsilon - \log 2) / \log ( \alpha +
  \sqrt{r/s} )$ iterations in the sense that
  \begin{equation}
    \label{eq:lsqr_y_converge}
    \|\hat{y}_{\CG} - y^* \|_{(AN)\T(AN)} \leq \varepsilon \| y^* \|_{(AN)\T(AN)}    
  \end{equation}
  holds with probability at least $1 - 2 e^{-\alpha^2 s/2}$ for any $\alpha \in
  ( 0, 1 - \sqrt{s/r} )$, where $\hat{y}_\CG$ is the approximate solution
  returned by the CG-like solver and $y^* = (AN)^\dagger b$. Let $\hat{x}_{\CG}
  = N \hat{y}_{\CG}$ be the approximate solution to the original problem. Since
  $x^* = N y^*$, \eqref{eq:lsqr_y_converge} is equivalent to
  \begin{equation}
    \label{eq:lsqr_x_converge}
    \|\hat{x}_{\CG} - x^* \|_{A\T A} \leq \varepsilon \| x^* \|_{A\T A},
  \end{equation}
  or in terms of residuals,
  \begin{equation}
    \label{eq:lsqr_r_converge}
    \|\hat{r}_{\CG} - r^*\|_2 \leq \varepsilon \| b - r^* \|_2,
  \end{equation}
  where $\hat{r}_{\CG} = b - A \hat{x}_{\CG}$ and $r^* = b - A x^*$.
\end{theorem}


In addition to allowing us to bound the number of iterations for CG-like
methods, the result given by \eqref{eq:sgm_concentration} also allows us to use
the CS method.  This method needs only one level-1
and two level-2 BLAS operations per iteration; and, importantly, because it
doesn't have vector inner products that require synchronization between nodes,
this method is suitable for clusters with high communication cost. It does need
an explicit bound on the singular values, but once that bound is tight, the CS
method has the same theoretical upper bound on the convergence rate as other
CG-like methods.  Unfortunately, in many cases, it is hard to obtain such an
accurate bound, which prevents the CS method becoming popular in practice.  In
our case, however, \eqref{eq:sgm_concentration} provides a probabilistic bound
with very high confidence.  Hence, we can employ the CS method without
difficulty.  For completeness, Algorithm \ref{alg:ls_chebyshev} describes the CS
method we implemented for solving LS problems.  For discussion of its
variations, see Gutknecht and Rollin~\cite{gutknecht2002chebyshev}.
\begin{algorithm}[h]
  \caption{Chebyshev semi-iterative (CS) method (computes $x \approx A^\dagger
    b$)}
  \label{alg:ls_chebyshev}
  \begin{algorithmic}[1]
    \STATE Given $A \in \mathbb{R}^{m \times n}$, $b \in \mathbb{R}^{m}$, and a
    tolerance $\varepsilon > 0$, choose $0 < \sigma_L \leq \sigma_U$ such that all
    non-zero singular values of $A$ are in $[\sigma_L, \sigma_U]$ and let $d =
    (\sigma_U^2+\sigma_L^2)/2$ and $c = (\sigma_U^2-\sigma_L^2)/2$.

    \STATE Let $x = \mathbf{0}$, $v = \mathbf{0}$, and $r = b$.

    \FOR{$k = 0, 1, \ldots, \left\lceil \left( \log \varepsilon - \log 2 \right) / \log \frac{ \sigma_U - \sigma_L }{ \sigma_U + \sigma_L } \right\rceil $} 

    \STATE $\beta \leftarrow
    \begin{cases}
      0 & \text{if } k = 0, \\
      \frac{1}{2} (c/d)^2 & \text{if } k = 1, \\
      (\alpha c/2)^2 & \text{otherwise,}
    \end{cases}$ 
    \quad 
    $\alpha \leftarrow
    \begin{cases}
      1/d & \text{if } k = 0, \\
      d-c^2/(2d) & \text{if } k = 1, \\
      1 / ( d - \alpha c^2 / 4 ) & \text{otherwise}.
    \end{cases}$

    \STATE $v \leftarrow \beta v + A\T r$

    \STATE $x \leftarrow x + \alpha v$

    \STATE $r \leftarrow r - \alpha A v$

    \ENDFOR
  \end{algorithmic}
\end{algorithm}

\subsection{Running time complexity}
\label{subsec:complexity}

In this section, we discuss the running time complexity of \texttt{LSRN}.  Let's
first calculate the computational cost of \texttt{LSRN}
(Algorithm~\ref{alg:ls_randn_tall}) in terms of floating-point operations
(flops).  Note that we need only $\tilde{\Sigma}$ and $\tilde{V}$ but not
$\tilde{U}$ or a full SVD of $\tilde{A}$ in step 4 of
Algorithm~\ref{alg:ls_randn_tall}. In step 6, we assume that the dominant cost
per iteration is the cost of applying $AN$ and $(AN)\T$. Then the total cost is
given by
\begin{eqnarray*}
  \label{eq:lsrn_flops}
  & s m  \times \text{flops}(\text{randn}) & \quad \text{for generating } G \\ 
  \null + \, & s \times \text{flops}(A\T u) & \quad \text{for computing } \tilde{A} \\
  \null + \, & 2 s n^2 + 11 n^3 & \quad \text{for computing } \tilde{\Sigma} \text{ and }  \tilde{V} \text{ \cite[p.\,254]{golub1996matrix}} \\
  \null + \, & N_{\text{iter}}  \times ( \text{flops}(A v) + \text{flops}(A\T u) + 4 n r) & \quad \text{for solving $\min_y \| A N y - b \|_2$},
\end{eqnarray*}
where lower-order terms are ignored. 
Here, $\text{flops}(\text{randn})$ is the average flop count to generate a 
sample from the standard normal distribution, while $\text{flops}(A v)$ and 
$\text{flops}(A\T u)$ are the flop counts for the respective matrix-vector 
products. 
If $A$ is a dense matrix, then we have 
$\text{flops}( A v ) = \text{flops}( A\T u ) = 2 m n$. 
Hence, the total cost becomes
\begin{equation*}
  \text{flops}(\texttt{LSRN}_{\text{dense}}) = s m \, \text{flops}(\text{randn}) + 2 s m n + 2 s n^2 + 11 n^3 + N_{\text{iter}} \times ( 4 m n + 4 n r ).
\end{equation*}
Comparing this with the SVD approach, which uses $2 m n^2 + 11 n^3$ flops, we
find \texttt{LSRN} requires more flops, even if we only consider computing
$\tilde{A}$ and its SVD.  However, the actual running time is not fully
characterized by the number of flops.  A matrix-matrix multiplication is much
faster than an SVD with the same number of flops.  We empirically compare the
running time in Section \ref{sec:experiments}.  If $A$ is a sparse matrix, we
generally have $\text{flops}(A v)$ and $\text{flops}(A\T u)$ of order
$\mathcal{O}(m)$. In this case, \texttt{LSRN} should run considerably faster
than the SVD approach. Finally, if $A$ is an operator, it is hard to apply SVD,
while \texttt{LSRN} still works without any modification. If we set $\gamma = 2$
and $\varepsilon = 10^{-14}$, we know $N_{\text{iter}} \approx 100$ by Theorem
\ref{thm:iter} and hence \texttt{LSRN} needs approximately $2n + 200$
matrix-vector multiplications.

One advantage of \texttt{LSRN} is that the stages of generating $G$ and 
computing $\tilde{A} = G A$ are embarrassingly parallel. 
Thus, it is easy to implement \texttt{LSRN} in parallel. 
For example, on a shared-memory machine using $p$ cores, the total running
time decreases~to
\begin{equation}
  \label{eq:mt_time}
  T_{\texttt{LSRN}}^{\text{mt},p} = T_{\text{randn}}/p + T_{\text{mult}}/p + T_{\text{svd}}^{\text{mt},p} + T_{\text{iter}}/p,
\end{equation}
where $T_{\text{randn}}$, $T_{\text{mult}}$, and $T_{\text{iter}}$ are the 
running times for the respective stages if \texttt{LSRN} runs on a single 
core, $T_{\text{svd}}^{\text{mt},p}$ is the running time of SVD using $p$ 
cores, and communication cost among threads is ignored. 
Hence, multi-threaded \texttt{LSRN} has very good scalability with 
near-linear speedup. 

Alternatively, let us consider a cluster of size $p$ using MPI, where each node
stores a portion of rows of $A$ (with $m \gg n$). Each node can generate random
samples and do the multiplication independently, and then an MPI\_Reduce
operation is needed to obtain $\tilde{A}$. Since $n$ is small, the SVD of
$\tilde{A}$ and the preconditioner $N$ are computed on a single node and
distributed to all the other nodes via an MPI\_Bcast operation. If the CS method
is chosen as the iterative solver, we need one MPI\_Allreduce operation per
iteration in order to apply $A\T$. Note that all the MPI operations that
\texttt{LSRN} uses are collective. If we assume the cluster is homogeneous and
has perfect load balancing, the time complexity to perform a collective
operation should be $\mathcal{O}(\log p)$. Hence the total running time becomes
\begin{equation}
  \label{eq:mpi_time}
  T_{\texttt{LSRN}}^{\text{mpi},p} = T_{\text{randn}}/p + T_{\text{mult}}/p + T_{\text{svd}} + T_{\text{iter}}/p + ( C_1 + C_2 N_{\text{iter}} ) \mathcal{O}( \log p ),
\end{equation}
where $C_1$ corresponds to the cost of computing $\tilde{A}$ and broadcasting
$N$, and $C_2$ corresponds to the cost of applying $A\T$ at each
iteration. Therefore, the MPI implementation of \texttt{LSRN} still has good
scalability as long as $T_{\text{svd}}$ is not dominant, i.e., as long as
$\tilde{A}$ is not too big.  Typical values of $n$ (or $m$ for under-determined
problems) in our empirical evaluations are around $1000$, and thus this is the
case.

\section{Tikhonov regularization}
\label{sec:regularization}

We point out that it is easy to extend \texttt{LSRN} to handle certain types of
Tikhonov regularization, also known as ridge regression.  Recall that Tikhonov
regularization involves solving the problem
\begin{equation}
  \label{eq:l2_reg}
  \text{minimize} \quad \frac{1}{2} \| A x - b \|_2^2 + \frac{1}{2} \| W x \|_2^2,
\end{equation}
where $W \in \mathbb{R}^{n \times n}$ controls the regularization term.  In many
cases, $W$ is chosen as $\lambda I_n$ for some value of a regularization
parameter $\lambda > 0$.  It is easy to see that \eqref{eq:l2_reg} is equivalent
to the following LS problem, without any regularization:
\begin{equation}
  \label{eq:l2_reg_eq}
  \text{minimize} \quad \frac{1}{2} \left\|
    \begin{pmatrix}
      A \\ W 
    \end{pmatrix}
    x - 
    \begin{pmatrix}
      b \\ 0
    \end{pmatrix} \right\|_2^2.
\end{equation}
This is an over-determined problem of size $(m+n) \times n$.  If $m \gg n$, then
we certainly have $m+n \gg n$.  Therefore, if $m \gg n$, we can directly apply
\texttt{LSRN} to \eqref{eq:l2_reg_eq} in order to solve \eqref{eq:l2_reg}.  On
the other hand, if $m \ll n$, then although \eqref{eq:l2_reg_eq} is still
over-determined, it is ``nearly square,'' in the sense that $m+n$ is only
slightly larger than $n$.  In this regime, random sampling methods and random
projection methods like \texttt{LSRN} do not perform well. In order to deal with
this regime, note that \eqref{eq:l2_reg} is equivalent to
\begin{eqnarray*}
  \text{minimize} &\quad& \frac{1}{2} \| r \|_2^2 + \frac{1}{2} \| W x \|_2^2 \\
  \text{subject to} &\quad& A x + r = b,
\end{eqnarray*}
where $r = b - A x$ is the residual vector.  (Note that we use $r$ to denote the
matrix rank in a scalar context and the residual vector in a vector context.)
By introducing $z = W x$ and assuming that $W$ is non-singular, we can re-write
the above problem as
\begin{eqnarray*}
  \text{minimize} &\quad& \frac{1}{2} \left\|
  \begin{pmatrix}
    z \\ r
  \end{pmatrix}
\right\|_2^2 \\
\text{subject to} &\quad& 
  \begin{pmatrix}
    A W^{-1} & I_m
  \end{pmatrix}
  \begin{pmatrix}
    z \\ r
  \end{pmatrix}
  =
  b,
\end{eqnarray*}
i.e., as computing the min-length solution to
\begin{equation}
  \label{eq:l2_reg_eq_under}
  \begin{pmatrix}
    A W^{-1} & I_m
  \end{pmatrix}
  \begin{pmatrix}
    z \\ r
  \end{pmatrix}
  =
  b.
\end{equation}
Note that \eqref{eq:l2_reg_eq_under} is an under-determined problem of size $m
\times (m+n)$. Hence, if $m \ll n$, we have $m \ll m+n$ and we can use
\texttt{LSRN} to compute the min-length solution to \eqref{eq:l2_reg_eq_under},
denoted by {\scriptsize $\begin{pmatrix} z^* \\ r^* \end{pmatrix}$}.  The solution to the
original problem \eqref{eq:l2_reg} is then given by $x^* = W^{-1} z^*$.  Here,
we assume that $W^{-1}$ is easy to apply, as is the case when $W=\lambda I_n$,
so that $A W^{-1}$ can be treated as an operator.
The equivalence between \eqref{eq:l2_reg} and \eqref{eq:l2_reg_eq_under} was 
first established by Herman, Lent, and Hurwitz~\cite{herman1980storage}.

In most applications of regression analysis, the amount of regularization, e.g.,
the optimal regularization parameter, is unknown and thus determined by
cross-validation.  This requires solving a sequence of LS problems where only
$W$ differs.  For over-determined problems, we only need to perform a random
normal projection on $A$ once.  The marginal cost to solve for each $W$ is the
following: a random normal projection on $W$, an SVD of size $\lceil \gamma n
\rceil \times n$, and a predictable number of iterations.  Similar results hold
for under-determined problems when each $W$ is a multiple of the identity
matrix.

\section{Numerical experiments}
\label{sec:experiments}

We implemented our LS solver \texttt{LSRN} and compared it with
competing solvers: LAPACK's DGELSD, MATLAB's backslash, and Blendenpik by Avron,
Maymounkov, and Toledo~\cite{avron2010blendenpik}.  MATLAB's backslash uses
different algorithms for different problem types. For sparse rectangular
systems, as stated by Tim
Davis\footnote{\url{http://www.cise.ufl.edu/research/sparse/SPQR/}},
``SuiteSparseQR \cite{davis2006direct,davis2008algorithm} is now QR in MATLAB
7.9 and $x=A \backslash b$ when $A$ is sparse and rectangular.''  Table
\ref{tab:lsq_solvers} summarizes the properties of those solvers.  We report our
empirical results in this section.

\begin{table}
  \centering
  \caption{LS solvers and their properties.}
  \begin{tabular}{c|c|c|c|c}
    \multirow{2}{*}{solver} & \multicolumn{2}{|c|}{min-len solution to} & \multicolumn{2}{|c}{taking advantage of} \\ 
    & under-det? & rank-def? & sparse $A$ & operator $A$ \\
    \hline
    LAPACK's DGELSD & yes & yes & no & no \\
    MATLAB's backslash & no & no & yes &  no \\
    Blendenpik & yes & no & no & no \\
    \texttt{LSRN} & yes & yes & yes & yes 
  \end{tabular}
  \label{tab:lsq_solvers}
\end{table}

\subsection{Implementation and system setup}

The experiments were performed on either a local shared-memory machine or a
virtual cluster hosted on Amazon's Elastic Compute Cloud (EC2). The
shared-memory machine has $12$ Intel Xeon CPU cores at clock rate 2GHz with
128GB RAM. The virtual cluster consists of $20$ m1.large instances configured by
a third-party tool called
StarCluster\footnote{\url{http://web.mit.edu/stardev/cluster/}}. An m1.large
instance has $2$ virtual cores with $2$ EC2 Compute Units\footnote{``One EC2
  Compute Unit provides the equivalent CPU capacity of a 1.0-1.2 GHz 2007
  Opteron or 2007 Xeon processor.'' from \url{http://aws.amazon.com/ec2/faqs/}}
each. To attain top performance on the shared-memory machine, we implemented a
multi-threaded version of \texttt{LSRN} in C, and to make our solver general
enough to handle large problems on clusters, we also implemented an MPI version
of \texttt{LSRN} in Python with NumPy, SciPy, and mpi4py. Both packages are
available for
download\footnote{\url{http://www.stanford.edu/group/SOL/software/lsrn.html}}. We
use the multi-threaded implementation to compare \texttt{LSRN} with other LS
solvers and use the MPI implementation to explore scalability and to compare
iterative solvers under a cluster environment. To generate values from the
standard normal distribution, we adopted the code from Marsaglia and
Tsang~\cite{marsaglia2000ziggurat} and modified it to use threads; this can
generate a billion samples in less than two seconds on the shared-memory
machine.  We also modified Blendenpik to call multi-threaded FFTW
routines. Blendenpik's default settings were used, i.e., using randomized
discrete Fourier transform and sampling $4 \min(m,n)$ rows/columns. All LAPACK's
LS solvers, Blendenpik, and \texttt{LSRN} are linked against MATLAB's own
multi-threaded BLAS and LAPACK libraries. So, in general, this is a fair setup
because all the solvers can use multi-threading automatically and are linked
against the same BLAS and LAPACK libraries.  The running times were
measured in wall-clock times.

\subsection{$\kappa(AN)$ and number of iterations}
\label{sec:condition-number} 

Recall that Theorem \ref{thm:cond_bound} states that $\kappa(AN)$, the condition
number of the preconditioned system, is roughly bounded by
$(1+\sqrt{r/s})/(1-\sqrt{r/s})$ when $s$ is large enough such that we can ignore
$\alpha$ in practice. To verify this statement, we generate random matrices of
size $10^4 \times 10^3$ with condition numbers ranged from $10^2$ to $10^8$. The
left figure in Figure \ref{fig:cond_and_iter} compares $\kappa(AN)$ with
$\kappa_+(A)$, the effective condition number of $A$, under different choices of
$s$ and $r$. We take the largest value of $\kappa(AN)$ in $10$ independent runs
as the $\kappa(AN)$ in the plot. For each pair of $s$ and $r$, the corresponding
estimate $(1+\sqrt{r/s})/(1-\sqrt{r/s})$ is drawn in a dotted line of the same
color, if not overlapped with the solid line of $\kappa(AN)$. We see that
$(1+\sqrt{r/s})/(1-\sqrt{r/s})$ is indeed an accurate estimate of the upper
bound on $\kappa(AN)$.  Moreover, $\kappa(AN)$ is not only independent of
$\kappa_+(A)$, but it is also quite small. For example, we have
$(1+\sqrt{r/s})/(1-\sqrt{r/s}) < 6$ if $s > 2 r$, and hence we can expect super
fast convergence of CG-like methods.
\begin{figure}
  \centering
  \includegraphics[width=0.48\textwidth]{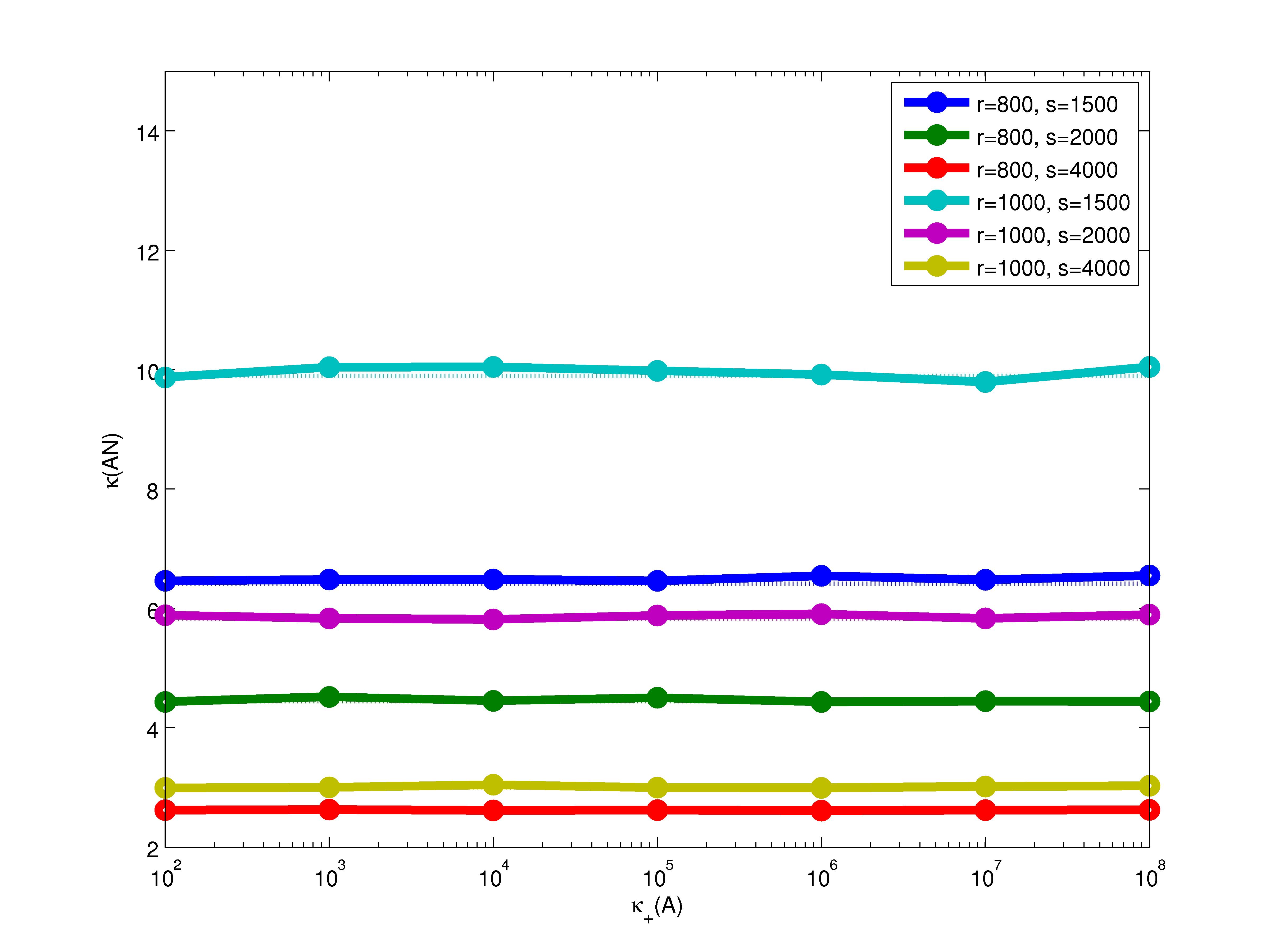}
  \includegraphics[width=0.48\textwidth]{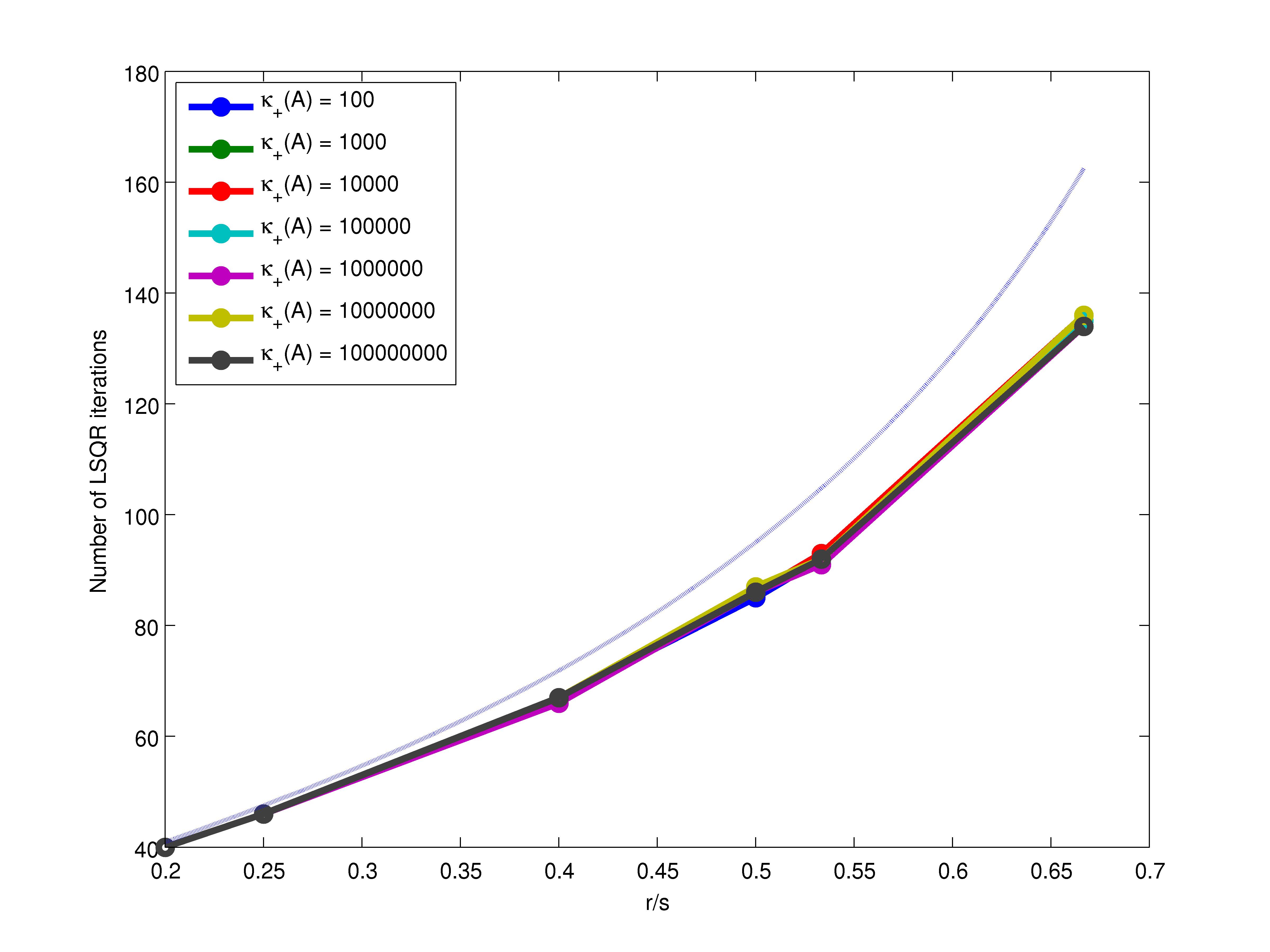}
  \caption{Left: $\kappa_+(A)$ vs.\ $\kappa(AN)$ for different choices of $r$
    and $s$. $A \in \mathbb{R}^{10^4 \times 10^3}$ is randomly generated with
    rank $r$. For each $(r,s)$ pair, we take the largest value of $\kappa(AN)$
    in 10 independent runs for each $\kappa_+(A)$ and connect them using a solid
    line. The estimate $(1+\sqrt{r/s})/(1-\sqrt{r/s})$ is drawn in a dotted line
    for each $(r,s)$ pair, if not overlapped with the corresponding solid
    line. Right: number of LSQR iterations vs.\ $r/s$. The number of LSQR
    iterations is merely a function of $r/s$, independent of the condition
    number of the original system.}
  \label{fig:cond_and_iter}
\end{figure}
Based on Theorem \ref{thm:iter}, the number of iterations should be less than
$(\log \varepsilon-\log 2)/\log \sqrt{r/s}$, where $\varepsilon$ is a given
tolerance. In order to match the accuracy of direct solvers, we set $\varepsilon
= 10^{-14}$. The right figure in Figure \ref{fig:cond_and_iter} shows the number
of LSQR iterations for different combinations of $r/s$ and $\kappa_+(A)$. Again,
we take the largest iteration number in $10$ independent runs for each pair of
$r/s$ and $\kappa_+(A)$. We also draw the theoretical upper bound $(\log
\varepsilon - \log 2)/\log \sqrt{r/s}$ in a dotted line. We see that the number of
iterations is basically a function of $r/s$, independent of $\kappa_+(A)$, and
the theoretical upper bound is very good in practice.  This confirms that the
number of iterations is fully predictable given $\gamma$.

\subsection{Tuning the oversampling factor $\gamma$}
\label{sec:tuning-parameters}

Once we set the tolerance and maximum number of iterations, there is only one
parameter left: the oversampling factor $\gamma$. To demonstrate the impact of
$\gamma$, we fix problem size to $10^5 \times 10^3$ and condition number to
$10^6$, set the tolerance to $10^{-14}$, and then solve the problem with
$\gamma$ ranged from $1.2$ to $3$. Figure \ref{fig:tuning-s} illustrates how
$\gamma$ affects the running times of \texttt{LSRN}'s stages: \texttt{randn} for
generating random numbers, \texttt{mult} for computing $\tilde{A} = G A$,
\texttt{svd} for computing $\tilde{\Sigma}$ and $\tilde{V}$ from $\tilde{A}$,
and \texttt{iter} for LSQR. We see that, the running times of \texttt{randn},
\texttt{mult}, and \texttt{svd} increase linearly as $\gamma$ increases, while
\texttt{iter} time decreases. Therefore there exists an optimal choice of
$\gamma$. For this particular problem, we should choose $\gamma$ between $1.8$
and $2.2$.
\begin{figure}
  \centering
  \includegraphics[width=0.48\textwidth]{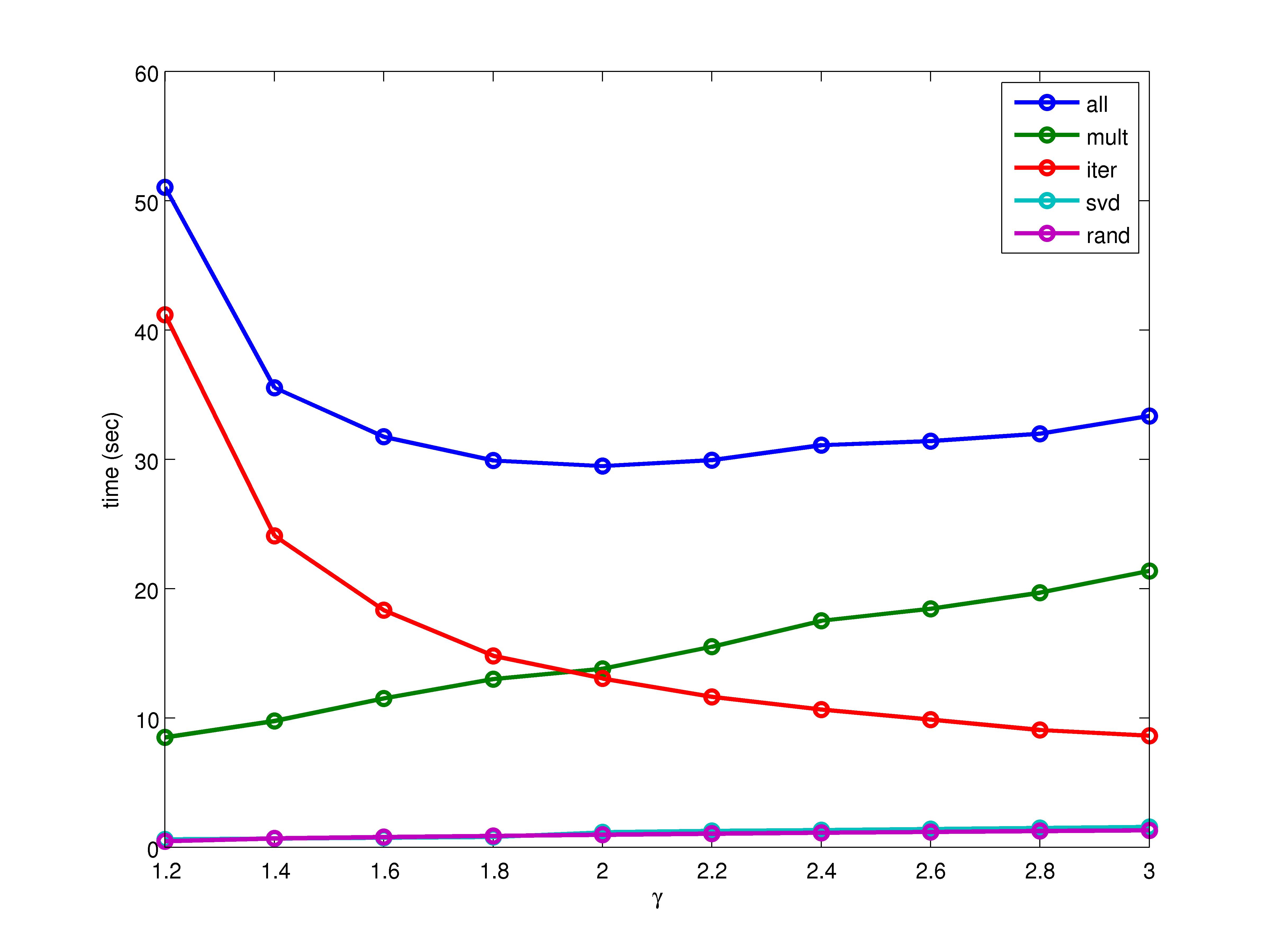}
  \caption{The overall running time of \texttt{LSRN} and the running time of
    each \texttt{LSRN} stage with different oversampling factor $\gamma$ for a
    randomly generated problem of size $10^5 \times 10^3$. For this particular
    problem, the optimal $\gamma$ that minimizes the overall running time lies
    in $[1.8, 2.2]$.}
  \label{fig:tuning-s}
\end{figure}
We experimented with various LS problems.  The best choice of $\gamma$ ranges
from $1.6$ to $2.5$, depending on the type and the size of the problem. We also
note that, when $\gamma$ is given, the running time of the iteration stage is
fully predictable. Thus we can initialize \texttt{LSRN} by measuring randn/sec
and flops/sec for matrix-vector multiplication, matrix-matrix multiplication,
and SVD, and then determine the best value of $\gamma$ by minimizing the total
running time \eqref{eq:mpi_time}.  For simplicity, we set $\gamma = 2.0$ in all
later experiments; although this is not the optimal setting for all cases, it is
always a reasonable choice.

\subsection{Dense least squares}
\label{sec:dense-lsq}

As the state-of-the-art dense linear algebra library, LAPACK provides several
routines for solving LS problems, e.g., DGELS, DGELSY, and
DGELSD. DGELS uses QR factorization without pivoting, which cannot handle
rank-deficient problems. DGELSY uses QR factorization with pivoting, which is
more reliable than DGELS on rank-deficient problems. DGELSD uses SVD. It is the
most reliable routine, and should be the most expensive as well. However, we
find that DGELSD actually runs much faster than DGELSY on strongly over- or
under-determined systems on the shared-memory machine. It may be because of
better use of multi-threaded BLAS, but we don't have a definitive explanation.

Figure \ref{fig:timing_dense_full_rank} compares the running times of
\texttt{LSRN} and competing solvers on randomly generated full-rank dense
strongly over- or under-determined problems. We set the condition numbers to
$10^6$ for all problems. Note that DGELS and DGELSD almost overlapped. The
results show that Blendenpik is the winner. For small-sized problems ($m \leq
3e4$), the follow-ups are DGELS and DGELSD. When the problem size goes larger,
\texttt{LSRN} becomes faster than DGELS/DGELSD. DGELSY is always slower than
DGELS/DGELSD, but still faster than MATLAB's backslash. The performance of
LAPACK's solvers decreases significantly for under-determined problems. We
monitored CPU usage and found that they couldn't fully use all the CPU cores,
i.e., they couldn't effectively call multi-threaded BLAS. Though still the best,
the performance of Blendenpik also decreases. \texttt{LSRN}'s performance does
not change much.

\begin{figure}
  \centering
  \includegraphics[width=0.48\textwidth]{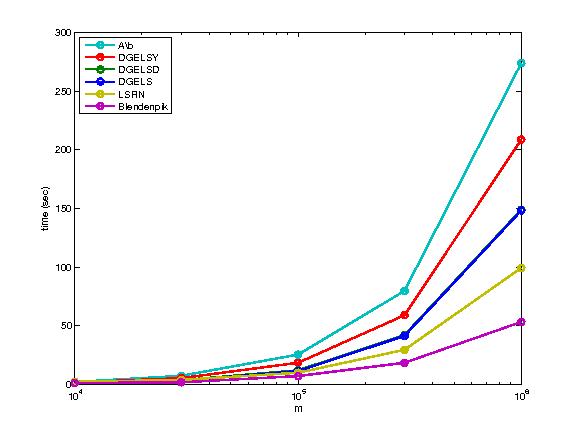}
  \includegraphics[width=0.48\textwidth]{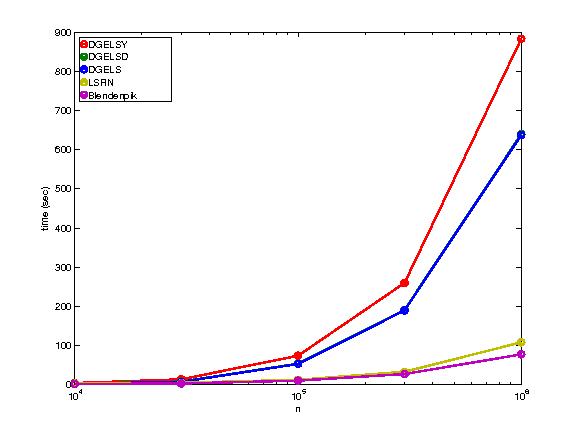}
  \caption{Running times on $m \times 1000$ dense over-determined problems with
    full rank (left) and on $1000 \times n$ dense under-determined problems with
    full rank (right).  Note that DGELS and DGELSD almost overlap. When $m >
    3e4$, we have Blendenpik $>$ \texttt{LSRN} $>$ DGELS/DGELSD $>$ DGELSY $>$
    $A\backslash b$ in terms of speed. On under-determined problems, LAPACK's
    performance decreases significantly compared with the over-determined
    cases. Blendenpik's performance decreases as well. \texttt{LSRN} doesn't
    change much.}
  \label{fig:timing_dense_full_rank}
\end{figure}

\texttt{LSRN} is also capable of solving rank-deficient problems, and in fact it
takes advantage of any rank-deficiency (in that it finds a solution in fewer
iterations).  Figure \ref{fig:timing_dense_rank_800} shows the results on over-
and under-determined rank-deficient problems generated the same way as in
previous experiments, except that we set $r = 800$. DGELSY and DGELSD remain the
same speed on over-determined problems as in full-rank cases, respectively, and
run slightly faster on under-determined problems. \texttt{LSRN}'s running times
reduce to $93$ seconds on the problem of size $10^6 \times 10^3$, from $100$
seconds on its full-rank counterpart.

\begin{figure}
  \centering
  \includegraphics[width=0.48\textwidth]{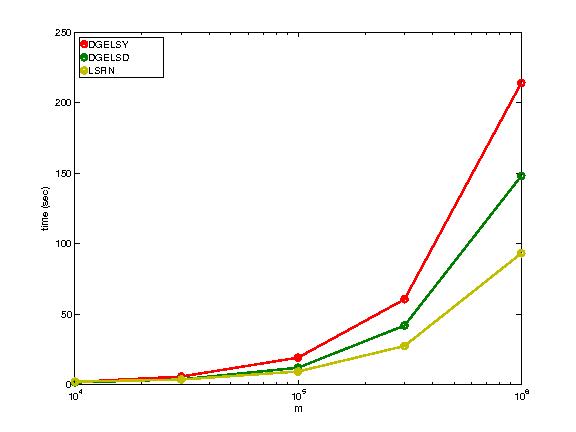}
  \includegraphics[width=0.48\textwidth]{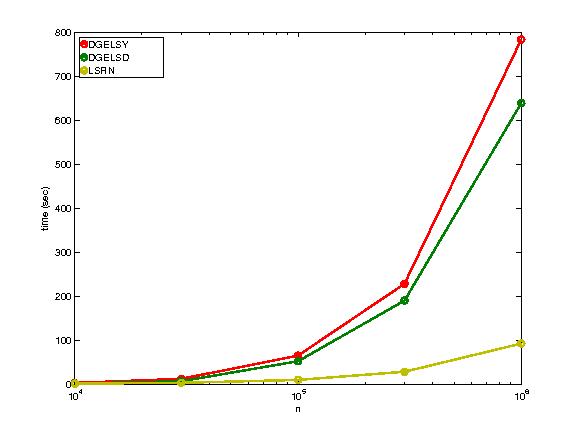}
  \caption{Running times on $m \times 1000$ dense over-determined problems with
    rank $800$ (left) and on $1000 \times n$ dense under-determined problems
    with rank $800$ (right). \texttt{LSRN} takes advantage of rank
    deficiency. We have \texttt{LSRN} $>$ DGSLS/DGELSD $>$ DGELSY in terms of
    speed.}
  \label{fig:timing_dense_rank_800}
\end{figure}

We see that, for strongly over- or under-determined problems, DGELSD is the
fastest and most reliable routine among the LS solvers provided by
LAPACK. However, it (or any other LAPACK solver) runs much slower on
under-determined problems than on over-determined problems, while \texttt{LSRN}
works symmetrically on both cases.  Blendenpik is the fastest dense least
squares solver in our tests. Though it is not designed for solving
rank-deficient problems, Blendenpik should be modifiable to handle such problems
following Theorem~\ref{thm:ls_precond_sufficient}.  We also note that
Blendenpik's performance depends on the distribution of the row norms of $U$. We
generate test problems randomly so that the row norms of $U$ are homogeneous,
which is ideal for Blendenpik. When the row norms of $U$ are heterogeneous,
Blendenpik's performance may drop.  See Avron, Maymounkov, and
Toledo~\cite{avron2010blendenpik} for a more detailed~analysis.

\subsection{Sparse least squares}
\label{sec:sparse-lsq}

In \texttt{LSRN}, $A$ is only involved in the computation of matrix-vector and
matrix-matrix multiplications. Therefore \texttt{LSRN} accelerates automatically
when $A$ is sparse, without exploring $A$'s sparsity pattern. LAPACK does not
have any direct sparse LS solver. MATLAB's backslash uses SuiteSparseQR by Tim
Davis \cite{davis2008algorithm} when $A$ is sparse and rectangular; this
requires explicit knowledge of $A$'s sparsity pattern to obtain a sparse QR
factorization.

We generated sparse LS problems using MATLAB's ``sprandn'' function with density
$0.01$ and condition number $10^6$. All problems have full rank. Figure
\ref{fig:timing_sparse_full_rank} shows the results on over-determined
problems. LAPACK's solvers and Blendenpik basically perform the same as in the
dense case. DGELSY is the slowest among the three. DGELS and DGELSD still
overlap with each other, faster than DGELSY but slower than Blendenpik. We see
that MATLAB's backslash handles sparse problems very well. On the $10^6 \times
10^3$ problem, backslash's running time reduces to $55$ seconds, from $273$
seconds on the dense counterpart. The overall performance of MATLAB's backslash
is better than Blendenpik's. \texttt{LSRN}'s curve is very flat. For small
problems ($m \leq 10^5$), \texttt{LSRN} is slow. When $m > 10^5$, \texttt{LSRN}
becomes the fastest solver among the six. \texttt{LSRN} takes only $23$ seconds
on the over-determined problem of size $10^6 \times 10^3$. On large
under-determined problems, \texttt{LSRN} still leads by a huge margin.
\begin{figure}
  \centering
  \includegraphics[width=0.48\textwidth]{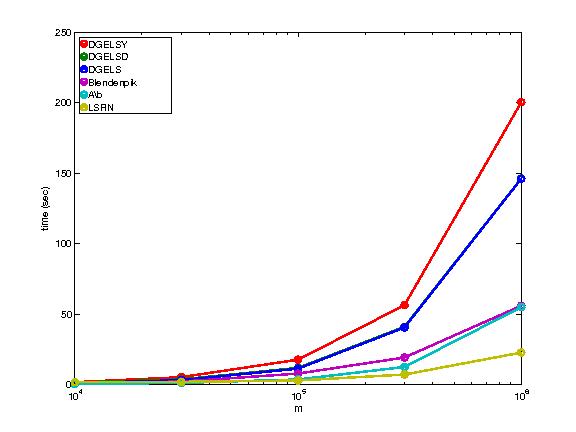}
  \includegraphics[width=0.48\textwidth]{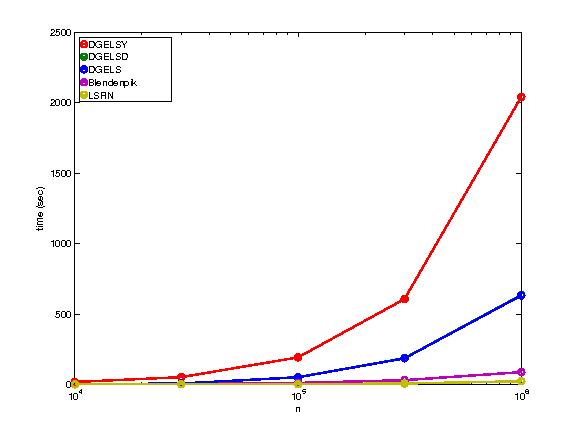}
  \caption{Running times on $m \times 1000$ sparse over-determined problems with
    full rank (left) and on $1000 \times n$ sparse under-determined problems
    with full rank (right). DGELS and DGELSD overlap with each other. LAPACK's
    solvers and Blendenpik perform almost the same as in the dense
    case. {\sc Matlab}'s backslash speeds up on sparse problems, and performs a
    little better than Blendenpik, but it is still slower than
    \texttt{LSRN}. \texttt{LSRN} leads by a huge margin on under-determined
    problems as well.}
  \label{fig:timing_sparse_full_rank}
\end{figure}

\texttt{LSRN} makes no distinction between dense and sparse problems.  The
speedup on sparse problems is due to faster matrix-vector and matrix-matrix
multiplications.  Hence, although no test was performed, we expect a similar
speedup on fast linear operators as well.  Also note that, in the multi-threaded
implementation of \texttt{LSRN}, we use a naive multi-threaded routine for
sparse matrix-vector and matrix-matrix multiplications, which is far from
optimized and thus leaves room for improvement.

\subsection{Real-world problems}
\label{sec:real-world-prob}

In this section, we report results on some real-world large data problems.  The
problems are summarized in Table \ref{tab:real-world-prob}, along with running
times.
\begin{table}
  \centering
  \caption{Real-world problems and corresponding running times in
    seconds. DGELSD doesn't take advantage of sparsity. Though MATLAB's backslash
    (SuiteSparseQR) may not give the min-length solutions to rank-deficient or
    under-determined problems, we still report its running times. Blendenpik either
    doesn't apply to rank-deficient problems or runs out of memory
    (OOM). \texttt{LSRN}'s running time is mainly determined by the problem size
    and the sparsity.}
  \scriptsize
  \newcommand{\z}{\phantom0}
  \begin{tabular}{l||c|c|c|r|c||r|r|c|c}
    matrix & $m$ & $n$ & nnz & rank & cond & DGELSD & $A \backslash b\ \ \ $ & Blendenpik & \texttt{LSRN} \\
    \hline
    \texttt{landmark} & 71952 & 2704 & 1.15e6 & 2671 & 1.0e8 & 29.54\z & 0.6498$^*$ &  - &  17.55 \\
    \texttt{rail4284} & 4284 & 1.1e6 & 1.1e7 & full & 400.0 & $>3600$\z & $1.203^*$ & OOM & 136.0 \\
    \hline
    \texttt{tnimg\_1} &\z951 & 1e6 & 2.1e7 & 925 & - & 630.6\z & $1067^*$ & - & 36.02 \\
    \texttt{tnimg\_2} & 1000 & 2e6 & 4.2e7 & 981 & - & 1291\z & $>3600^*$ & - & 72.05 \\
    \texttt{tnimg\_3} & 1018 & 3e6 & 6.3e7 & 1016 & - & 2084\z & $>3600^*$ & - & 111.1 \\
    \texttt{tnimg\_4} & 1019 & 4e6 & 8.4e7 & 1018 & - &  2945\z & $>3600^*$ & - & 147.1 \\
    \texttt{tnimg\_5} & 1023 & 5e6 & 1.1e8 & full &  - & $>3600$\z & $>3600^*$ & OOM & 188.5 \\
  \end{tabular}
  \label{tab:real-world-prob}
\end{table}

\texttt{landmark} and \texttt{rail4284} are from the University of Florida
Sparse Matrix Collection \cite{davis1997university}. \texttt{landmark}
originated from a rank-deficient LS problem. \texttt{rail4284} has full rank and
originated from a linear programming problem on Italian railways. Both matrices
are very sparse and have structured patterns. MATLAB's backslash (SuiteSparseQR)
runs extremely fast on these two problems, though it doesn't guarantee to return
the min-length solution. Blendenpik is not designed to handle the rank-deficient
\texttt{landmark}, and it unfortunately runs out of memory (OOM) on
\texttt{rail4284}. \texttt{LSRN} takes 17.55 seconds on \texttt{landmark} and
136.0 seconds on \texttt{rail4284}. DGELSD is slightly slower than \texttt{LSRN}
on \texttt{landmark} and much slower on \texttt{rail4284}.

\texttt{tnimg} is generated from the TinyImages collection
\cite{torralba2008tiny}, which provides $80$ million color images of size $32
\times 32$. For each image, we first convert it to grayscale, compute its
two-dimensional DCT, and then only keep the top $2\%$ largest coefficients in
magnitude. This gives a sparse matrix of size $1024 \times 8\mathrm{e}7$ where
each column has $20$ or $21$ nonzero elements. Note that \texttt{tnimg} doesn't
have apparent structured pattern. Since the whole matrix is too big, we work on
submatrices of different sizes. \texttt{tnimg}\_$i$ is the submatrix consisting
of the first $10^6 \times i$ columns of the whole matrix for $i=1,\ldots,80$,
where empty rows are removed. The running times of \texttt{LSRN} are
approximately linear in $n$. Both DGELSD and MATLAB's backslash are very slow on
the \texttt{tnimg} problems. Blendenpik either doesn't apply to the
rank-deficient cases or runs OOM.

We see that, though both methods taking advantage of sparsity, MATLAB's
backslash relies heavily on the sparsity pattern, and its performance is
unpredictable until the sparsity pattern is analyzed, while \texttt{LSRN}
doesn't rely on the sparsity pattern and always delivers predictable performance
and, moreover, the min-length solution.

\subsection{Scalability and choice of iterative solvers on clusters}
\label{sec:scalability}

In this section, we move to the Amazon EC2 cluster. The goals are to
demonstrate that
\begin{inparaenum} 
\item[(1)] \texttt{LSRN} scales well on clusters, and
\item[(2)] the CS method is preferred to LSQR on clusters
  with high communication cost.
\end{inparaenum}
The test problems are submatrices of the \texttt{tnimg} matrix in the previous
section: \texttt{tnimg}\_4, \texttt{tnimg}\_10, \texttt{tnimg}\_20, and
\texttt{tnimg}\_40, solved with $4$, $10$, $20$, and $40$ cores
respectively. Each process stores a submatrix of size $1024 \times
1\mathrm{e}6$. Table \ref{tab:cluster} shows the results, averaged over $5$
runs.
\begin{table}
  \centering
  \caption{Test problems on the Amazon EC2 cluster and corresponding running
    times in seconds.  When we enlarge the problem scale by a factor of $10$ and
    increase the number of cores accordingly, the running time only increases by a
    factor of $50\%$. It shows \texttt{LSRN}'s good scalability.  Though the CS
    method takes more iterations, it is faster than LSQR by saving communication
    cost.}
  \small
  \newcommand{\z}{\phantom0}
  \begin{tabular}{l|c|c|l|c|c|c|c|c|c}
    \ \ \ \ \ solver & $N_{\text{nodes}}$ & np & matrix & $m$ & $n$ & nnz &  $N_{\text{iter}}$ & $T_{\text{iter}}$ & $T_{\text{total}}$ \\
    \hline
    \texttt{LSRN} w/ CS & \multirow{2}{*}{\z2} & \multirow{2}{*}{\z4} & \multirow{2}{*}{\texttt{tnimg}\_4} & \multirow{2}{*}{1024} & \multirow{2}{*}{4e6} & \multirow{2}{*}{8.4e7} & 106 & 34.03 & 170.4 \\
    \texttt{LSRN} w/ LSQR & & & & & & & \z84 & 41.14 & 178.6 \\
    \hline
    \texttt{LSRN} w/ CS & \multirow{2}{*}{\z5} & \multirow{2}{*}{10} &  \multirow{2}{*}{\texttt{tnimg}\_10} & \multirow{2}{*}{1024} & \multirow{2}{*}{1e7} & \multirow{2}{*}{2.1e8} & 106 & 50.37 & 193.3 \\
    \texttt{LSRN} w/ LSQR & & & & & & & \z84 & 68.72 & 211.6 \\
    \hline
    \texttt{LSRN} w/ CS & \multirow{2}{*}{10} & \multirow{2}{*}{20} &  \multirow{2}{*}{\texttt{tnimg}\_20} & \multirow{2}{*}{1024} & \multirow{2}{*}{2e7} & \multirow{2}{*}{4.2e8}   & 106 & 73.73 & 220.9 \\
    \texttt{LSRN} w/ LSQR & &  & &  & & & \z84 & 102.3 & 249.0 \\
    \hline
    \texttt{LSRN} w/ CS &  \multirow{2}{*}{20} & \multirow{2}{*}{40} &  \multirow{2}{*}{\texttt{tnimg}\_40} & \multirow{2}{*}{1024} & \multirow{2}{*}{4e7} & \multirow{2}{*}{8.4e8}  & 106 & 102.5 & 255.6 \\
    \texttt{LSRN} w/ LSQR & & & &  &  & & \z84 & 137.2 & 290.2 \\
  \end{tabular}
  \label{tab:cluster}
\end{table}
Ideally, from the complexity analysis \eqref{eq:mpi_time}, when we double $n$
and double the number of cores, the increase in running time should be a
constant if the cluster is homogeneous and has perfect load balancing (which we
have observed is not true on Amazon EC2).  For \texttt{LSRN} with CS, from
\texttt{tnimg}\_10 to \texttt{tnimg}\_20 the running time increases $27.6$
seconds, and from \texttt{tnimg}\_20 to \texttt{tnimg}\_40 the running time
increases $34.7$ seconds.  We believe the difference between the time increases
is caused by the heterogeneity of the cluster, because Amazon EC2 doesn't
guarantee the connection speed among nodes.  From \texttt{tnimg}\_4 to
\texttt{tnimg}\_40, the problem scale is enlarged by a factor of $10$ while the
running time only increases by a factor of $50\%$. The result still demonstrates
\texttt{LSRN}'s good scalability. We also compare the performance of LSQR and CS
as the iterative solvers in \texttt{LSRN}. For all problems LSQR converges in
$84$ iterations and CS converges in $106$ iterations. However, LSQR is slower
than CS. The communication cost saved by CS is significant on those tests. As a
result, we recommend CS as the default \texttt{LSRN} iterative solver for
cluster environments. Note that to reduce the communication cost on a cluster,
we could also consider increasing $\gamma$ to reduce the number of iterations.

\section{Conclusion}
\label{sec:conclusion}

We developed \texttt{LSRN}, a parallel solver for strongly over- or
under-determined, and possibly rank-deficient, systems.  \texttt{LSRN} uses
random normal projection to compute a preconditioner matrix for an iterative
solver such as LSQR and the Chebyshev semi-iterative (CS) method. The
preconditioning process is embarrassingly parallel and automatically speeds up
on sparse matrices and fast linear operators, and on rank-deficient data.  We
proved that the preconditioned system is consistent and extremely
well-conditioned, and derived strong bounds on the number of iterations of LSQR
or the CS method, and hence on the total running time. On large dense systems,
\texttt{LSRN} is competitive with the best existing solvers, and it runs
significantly faster than competing solvers on strongly over- or
under-determined sparse systems. \texttt{LSRN} is easy to implement using
threads or MPI, and it scales well in parallel environments.

\section*{Acknowledgements}

After completing the initial version of this manuscript, we learned of
the LS algorithm of Coakley et al.\ \cite{coakley2011fast}.
We thank Mark Tygert for pointing us to this reference.  We are also
grateful to Lisandro Dalcin, the author of mpi4py, for his own version
of the MPI\_Barrier function to prevent idle processes from
interrupting the multi-threaded SVD process too frequently.

\newpage

\bibliographystyle{siam}
\bibliography{lsrn}

\end{document}